\documentclass[pdflatex,sn-mathphys-num]{sn-jnl}

\usepackage{graphicx}%
\usepackage{multirow}%
\usepackage{amsmath,amssymb,amsfonts}%
\usepackage{amsthm}%
\usepackage{mathrsfs}%
\usepackage{xcolor}%
\usepackage{textcomp}%
\usepackage{booktabs}%
\usepackage{algorithm}%
\usepackage{algorithmicx}%
\usepackage{algpseudocode}%
\usepackage{listings}%
\usepackage{mathtools}
\usepackage{array}
\usepackage{makecell}
\usepackage{multicol}
\usepackage{adjustbox}
\usepackage{marginnote}
\usepackage{tikz}
\usepackage{xurl}
\usepackage{longtable}
\usepackage{fancyvrb}
\usepackage{pgfplots}
\usepackage{float}
\usepackage{seqsplit}
\usepackage{siunitx}
\usepackage{needspace}
\pgfplotsset{compat=1.18}
\usetikzlibrary{plotmarks,matrix,positioning}

\theoremstyle{thmstyleone}
\newtheorem{theorem}{Theorem}
\newtheorem{lemma}[theorem]{Lemma}
\newtheorem{proposition}[theorem]{Proposition}
\newtheorem{corollary}[theorem]{Corollary}
\theoremstyle{thmstyletwo}

\definecolor{alert-red-bg}{rgb}{1, 0.85, 0.85}
\definecolor{alert-red-border}{rgb}{0.8, 0, 0}
\definecolor{warning-yellow-bg}{rgb}{1, 1, 0.8}
\definecolor{warning-yellow-border}{rgb}{0.8, 0.8, 0}
\definecolor{info-blue-bg}{rgb}{0.85, 0.9, 1}
\definecolor{info-blue-border}{rgb}{0, 0, 0.8}

\algnewcommand{\LineComment}[1]{\State $\triangleright$ #1}

\def\00{{\bf 0}}\def\11{{\bf 1}}\def\+{\oplus}

\providecommand{\newoperator}[3]{\newcommand*{#1}{\mathop{#2}#3}}
\newoperator{\FD}{\mathrm{FD}}{\nolimits}
\newlength{\sqsize}
\definecolor{violet}{rgb}{0.53, 0.0, 0.69}

\begin{document}

\title[Extended Differential Cryptanalysis of Kuznyechik]{Extended Differential Cryptanalysis of Kuznyechik}

\author[1]{\fnm{Pantelimon} \sur{St\u{a}nic\u{a}}}\email{pstanica@nps.edu}
\author[2]{\fnm{Ranit} \sur{Dutta}}\email{p22ma207@iitj.ac.in}
\author[2]{\fnm{Bimal} \sur{Mandal}}\email{bimalmandal@iitj.ac.in}

\affil[1]{Applied Mathematics Department, Naval Postgraduate School, Monterey, CA 93943, USA}
\affil[2]{Department of Mathematics, Indian Institute of Technology Jodhpur, Jodhpur, Rajasthan 342030, India}

\abstract{
We study the input-side relation $F(cx\oplus a)\oplus F(x)=b$ as an extended differential for block-cipher analysis. Because multiplication occurs before the first S-box, fixing that S-box output leaves an ordinary XOR difference for subsequent binary linear layers and common key additions. For permutations, the extended and outer $c$-differential tables are related by inversion.

For the Kuznyechik S-box, the exceptional inverse classes $\mathtt{02}/\mathtt{e1}$, $\mathtt{04}/\mathtt{91}$, and $\mathtt{03}/\mathtt{be}$ have exact extended differential uniformities $64$, $33$, and $21$. We connect them to the pseudo-exponential structure of Perrin and Udovenko through a cross-field transfer theorem. In the common polynomial basis, the disagreement rank between multiplication maps in the hidden and Kuznyechik fields equals the degree of the multiplier polynomial, yielding structural lower bounds $60$, $30$, and $16$. We also analyze the complete extended differential table: its normalization is doubly stochastic, its energy is an exact multiplicative correlation of ordinary DDT rows, and transferred hidden fibers give explicit Pearson and Shannon-information bounds.

Whitening selects the effective row $a\oplus(1\oplus c)K$. For two-round standard Kuznyechik we derive the exact fixed-key channel $Q_{c,\ell}=W_cP_\ell W_1$ and a two-stage maximum-likelihood attack. With $c=\mathtt{04}$ it recovers the full $256$-bit master key with failure probability at most $0.01$, using $226{,}512$ chosen-plaintext pairs and at most $240{,}669$ encryption queries. Restricted product-model searches give gains of $5.2$, $4.6$, and about $1.6$ bits over the corresponding classical families at two, three, and four rounds. A separate nine-round no-pre-whitening Monte Carlo study is reported only as exploratory evidence because its final configurations followed preliminary screening.
}


\maketitle

\section{Introduction}

Differential cryptanalysis combines local transition biases across rounds~\cite{BS91,BS92}. For XOR differences, binary linear layers propagate differences deterministically and common round-key additions cancel. Outer $c$-differentials instead use $F(x\oplus a)\oplus cF(x)=b$~\cite{Ellingsen2020}; their direct use in a conventional binary SPN is less natural because the multiplicative relation at the S-box output need not survive later linear and key-addition layers~\cite{BartoliKolsMicheli2025,Baudrin23,BaudrinGeneral2025}.
We therefore use the input-side relation
\begin{equation}
F(cx\oplus a)\oplus F(x)=b,
\label{eq:intro-inner}
\end{equation}
which we call an \emph{extended differential}. Once the first S-box output difference $b$ is fixed, subsequent binary linear layers and common key additions act on an ordinary XOR difference. Input-side multiplication has appeared in multiplicative-differential work~\cite{BJW02}; here the purpose is specifically to confine the nonclassical relation to the first nonlinear interface. For permutations, the extended and outer differential tables are linked by inversion, ${}_c\Delta_F(a,b)=\nabla_{c,F^{-1}}(b,a)$. Equation~\eqref{eq:intro-inner} can also be viewed as a local commutative-diagram relation in the sense of Wagner and later commutative cryptanalysis~\cite{Wagner2004Comm,Baudrin23,BaudrinGeneral2025}.

Kuznyechik provides a striking example. Its S-box has ordinary differential uniformity $8$, whereas the extended differential spectrum contains the exceptional inverse classes $\mathtt{02}/\mathtt{e1}$, $\mathtt{04}/\mathtt{91}$, and $\mathtt{03}/\mathtt{be}$ with values $64$, $33$, and $21$. We relate these values to the hidden algebraic structure identified in the reverse-engineering and pseudo-exponential analyses of the Kuznyechik/Streebog S-box~\cite{BPU16,PerrinUdovenko2016,PerrinTKlog2019}. The hidden representation and the public Kuznyechik field use different irreducible polynomials on the same binary vector space. A hidden multiplicative transition therefore splits into affine extended-differential offset classes in the public representation. We formalize this by a cross-field transfer theorem and prove that, in the common polynomial basis, the rank of the disagreement between the two multiplication maps is exactly the degree of the multiplier polynomial. This yields structural lower bounds $60$, $30$, and $16$ for the three exceptional classes.

Whitening makes the complete extended differential table cryptographically relevant. If plaintext bytes satisfy $x'=cx\oplus a$, then after whitening by $K$ the effective row is
$A=a\oplus(1\oplus c)K$. We therefore analyze the complete table $N_c(a,b)$ rather than only its maximum entry: its normalization is doubly stochastic, its energy is an exact multiplicative correlation of ordinary DDT rows, and the same quantity controls Pearson information about a uniform whitening byte. Parallel factorization yields, among other consequences, a structured $2^{-7}$ master-key class for a specified two-round event.

For two-round standard Kuznyechik, including whitening and the RFC~7801 key schedule, we derive the exact fixed-key channel $Q_{c,\ell}=W_cP_\ell W_1$. It does not use a Markov approximation. With the structurally explained choice $c=\mathtt{04}$, a two-stage maximum-likelihood attack recovers the full $256$-bit master key with failure probability at most $0.01$, using $226{,}512$ chosen-plaintext pairs and at most $240{,}669$ encryption queries. Separate restricted product-model searches show improvements of $5.2$, $4.6$, and about $1.6$ bits over the corresponding classical families at two, three, and four rounds. Finally, we report a no-pre-whitening long-round Monte Carlo study only as exploratory evidence because the final configurations followed preliminary screening.

\paragraph{Contributions}
We (i) relate the extended and outer $c$-differential orientations and isolate the nonclassical relation to the first nonlinear interface; (ii) explain the exceptional Kuznyechik multipliers by a cross-field transfer theorem and an exact degree--rank result; (iii) derive complete-table energy and information identities, including the whitening interpretation and a structured $2^{-7}$ key class; (iv) give an exact two-round fixed-key channel and full $256$-bit master-key recovery; and (v) compare restricted two--four-round characteristic families and report a clearly separated exploratory long-round experiment. Detailed tables, logs, and auxiliary computations are placed in the supplementary material and reproducibility repository.

\section{Preliminaries}
\label{sec-prel}

Throughout the paper, let $q=2^n$ and let
$F:\mathbb F_q\rightarrow\mathbb F_q$ be a permutation. Addition in $\mathbb F_q$ is denoted by $\oplus$. For the Kuznyechik S-box we have $q=2^8$. Bytes are identified with elements of $\mathbb F_{2^8}$ using the polynomial basis associated with $x^8+x^7+x^6+x+1$.

\subsection{Differential and outer $c$-differential notation}

For $\alpha,\beta\in\mathbb F_q$, let
\[
D_F(\alpha,\beta)
=
\#\{x\in\mathbb F_q:
F(x\oplus\alpha)\oplus F(x)=\beta\}
\]
denote the ordinary difference distribution table (DDT) of $F$. Its
differential uniformity is
$\delta(F)=\max_{\alpha\ne0,\beta}D_F(\alpha,\beta)$.
Ellingsen et al.~\cite{Ellingsen2020} introduced the $c$-differential
table
\begin{equation}
{}_c\Delta_F(a,b)
=
\#\{x\in\mathbb F_q:
F(x\oplus a)\oplus cF(x)=b\}.
\label{eq:outer-cddt}
\end{equation}
For fixed $c$, the corresponding outer $c$-differential uniformity is the
largest value of ${}_c\Delta_F(a,b)$, with $a\ne0$ in the classical case
$c=1$. We refer to~\eqref{eq:outer-cddt} as the \emph{outer}
$c$-differential in order to distinguish it from the input-side
orientation introduced in Section~\ref{sec-ccd}.

\subsection{Kuznyechik}

Kuznyechik~\cite{RFC7801} is a $128$-bit block cipher with a $256$-bit master key. Its
state consists of $16$ bytes. We write $S$ for the Kuznyechik byte S-box
and, when the argument is a state, also for its parallel application to all
$16$ bytes. The nonlinear layer is followed by the linear transformation
$L$. The byte S-box is a permutation with ordinary differential uniformity
$8$.
The byte-oriented linear layer has branch number $17$~\cite{AY15}.
Hence, for every nonzero state difference $\Delta$,
\begin{equation}
\operatorname{wt}(\Delta)+\operatorname{wt}(L(\Delta))\ge17,
\label{eq:kuz-branch}
\end{equation}
where $\operatorname{wt}$ denotes the number of nonzero bytes.
By~\eqref{eq:kuz-branch}, if $L(\Delta)$ has byte weight one, then all $16$ bytes of $\Delta$ are nonzero. This observation is used later in
the reduced-round characteristic analysis.

The paper considers two closely related settings. The exact
key-recovery analysis uses reduced-round \emph{standard} Kuznyechik and
therefore retains the initial key addition. The statistical and
characteristic studies also consider the explicitly stated
no-pre-whitening core. We distinguish these settings whenever they are
used and do not identify probabilities obtained in one model with
those of the other.

The complete S-box, linear transformation, key schedule, and encryption
algorithm are given in the supplementary material.

\section{Extended differentials}
\label{sec-ccd}

The outer $c$-differential places the multiplier at the output of
the S-box. For block-cipher applications, we instead place it at the input.
For $c\in\mathbb F_q^*$ and $a,b\in\mathbb F_q$, define
\begin{equation}
\nabla_{c,F}(a,b)
=
\#\{x\in\mathbb F_q:
F(cx\oplus a)\oplus F(x)=b\}.
\label{eq:inner-cddt}
\end{equation}
We call $\nabla_{c,F}$ the \emph{extended differential table} of $F$.
Its uniformity is
$\delta_c(F)=\max_{a,b}\nabla_{c,F}(a,b)$, with the usual restriction
$a\ne0$ when $c=1$. Thus $c=1$ gives the ordinary DDT\@. When $F$ is fixed, we write $N_c(a,b)=\nabla_{c,F}(a,b)$. Unless stated otherwise, $N_c$ denotes the full $q\times q$ table, including the row $a=0$; the restriction $a\ne0$ applies only when taking the classical uniformity at $c=1$.

Input-side multiplication is related to the multiplicative-differential
viewpoint of Borisov et al.~\cite{BJW02}. With
$A_{c,a}(x)=cx\oplus a$ and $T_b(y)=y\oplus b$, every solution counted by
\eqref{eq:inner-cddt} satisfies $F(A_{c,a}(x))=T_b(F(x))$. Thus the table
counts the inputs on which this affine diagram commutes; it does not assume
that $F\circ A_{c,a}=T_b\circ F$ holds identically. This places the relation
within the affine/commutative viewpoint~\cite{Wagner2004Comm,
Baudrin23,BaudrinGeneral2025}. Our use is more specific: after the first
S-box, we require the two paths to differ by an ordinary XOR difference.

\subsection{Relation with the outer $c$-differential}

For permutations, the extended and outer orientations are directly related
through inversion.

\begin{theorem}
\label{thm:inner-outer-inverse}
Let $F$ be a permutation of $\mathbb F_q$. For every
$c\in\mathbb F_q^*$ and $a,b\in\mathbb F_q$,
\[
{}_c\Delta_F(a,b)
=
\nabla_{c,F^{-1}}(b,a).
\]
\end{theorem}

\begin{proof}
The equation
$F(x\oplus a)\oplus cF(x)=b$ is equivalent to
$F(x\oplus a)=cF(x)\oplus b$.
Put $y=F(x)$. Since $F$ is a permutation, this gives
$x=F^{-1}(y)$ and
$x\oplus a=F^{-1}(cy\oplus b)$. Hence
\[
F^{-1}(cy\oplus b)\oplus F^{-1}(y)=a.
\]
The number of solutions is therefore
$\nabla_{c,F^{-1}}(b,a)$.
\end{proof}

Theorem~\ref{thm:inner-outer-inverse} shows that the extended orientation
does not define an unrelated differential notion. It is the input-side
counterpart of the outer $c$-differential applied to the inverse
permutation.

\subsection{The first nonlinear interface}

The reason for using the extended orientation is its behavior after the
first S-box. Suppose
$F(cx\oplus a)\oplus F(x)=b$. If $L$ is any
$\mathbb F_2$-linear map, then
$L(F(cx\oplus a))\oplus L(F(x))=L(b)$. The same statement remains true
after adding the same round key to both paths.

\begin{lemma}
\label{lem:first-interface}
Let $A(y)=L(y)\oplus k$, where $L$ is $\mathbb F_2$-linear and $k$ is
constant. If
$F(cx\oplus a)\oplus F(x)=b$, then
\[
A(F(cx\oplus a))\oplus A(F(x))=L(b).
\]
\end{lemma}

\begin{proof}
The key term cancels, and linearity gives
\[
A(F(cx\oplus a))\oplus A(F(x))
=
L\bigl(F(cx\oplus a)\oplus F(x)\bigr)
=
L(b).
\]
\end{proof}

Thus the multiplier $c$ is needed only at the first nonlinear
interface. Once the S-box output difference $b$ is fixed, the two
states are related by an ordinary XOR difference. Later linear layers
and key additions can therefore be handled in the usual differential
setting.

\subsection{The Kuznyechik observation}
\label{subsec:kuz-inner-observation}

We now specialize to the Kuznyechik S-box, denoted by $S$. Its ordinary
differential uniformity is $8$. The extended differential spectrum is
very different for a small number of multipliers. The most prominent
values are summarized in Table~\ref{tab:kuz-inner-selected}.

\begin{table}[t]
\centering
\begin{minipage}{0.72\textwidth}
\centering
\caption{Selected extended differential uniformities of the Kuznyechik S-box.
The complete spectrum is given in the supplementary material.}
\label{tab:kuz-inner-selected}
\footnotesize
\begin{tabular*}{0.75\textwidth}
  {@{\extracolsep{\fill}}cc@{}}
\toprule
Multiplier class & $\delta_c(S)$ \\
\midrule
$\mathtt{01}$ & $8$ \\
$\mathtt{02}/\mathtt{e1}$ & $64$ \\
$\mathtt{04}/\mathtt{91}$ & $33$ \\
$\mathtt{03}/\mathtt{be}$ & $21$ \\
\bottomrule
\end{tabular*}
\end{minipage}
\end{table}

The appearance of inverse pairs is structural. It follows from a simple
symmetry of the complete extended differential table.

\begin{proposition}
\label{prop:multiplier-inversion}
For every $c\in\mathbb F_q^*$ and $a,b\in\mathbb F_q$,
\[
\nabla_{c,F}(a,b)
=
\nabla_{c^{-1},F}(c^{-1}a,b).
\]
Consequently,
$\delta_c(F)=\delta_{c^{-1}}(F)$.
\end{proposition}

\begin{proof}
Starting from
$F(cx\oplus a)\oplus F(x)=b$, put $y=cx\oplus a$.
Then $x=c^{-1}y\oplus c^{-1}a$, and the equation becomes
\[
F(c^{-1}y\oplus c^{-1}a)\oplus F(y)=b.
\]
This is the extended differential relation with multiplier $c^{-1}$ and offset $c^{-1}a$.
\end{proof}

The same row-relabeling symmetry also survives the byte projections used in the statistical experiment.

\begin{proposition}
\label{prop:masked-inverse}
Let $V=(\mathbb F_q)^s$, let $E:V\to V$ be any permutation, let multiplication by $c\in\mathbb F_q^*$ act componentwise, and let $\pi_O:V\to W$ be a fixed output projection. Define
\[
N^{\pi_O}_{c,E}(a,b)=\#\{x\in V:\pi_O(E(cx\oplus a)\oplus E(x))=b\}.
\]
Then
\[
N^{\pi_O}_{c,E}(a,b)=N^{\pi_O}_{c^{-1},E}(c^{-1}a,b).
\]
If the allowed input-offset space is invariant under nonzero scalar multiplication, as are the byte-aligned spaces used later, the complete masked tables for $c$ and $c^{-1}$ are row permutations of one another.
\end{proposition}

\begin{proof}
Put $y=cx\oplus a$. Then $x=c^{-1}y\oplus c^{-1}a$, and
\[
E(cx\oplus a)\oplus E(x)=E(y)\oplus E(c^{-1}y\oplus c^{-1}a).
\]
Applying $\pi_O$ and using commutativity of XOR gives the claimed identity. Multiplication by $c^{-1}$ permutes every coordinate subspace closed under nonzero scalar multiplication.
\end{proof}

In the Kuznyechik field,
$\mathtt{02}^{-1}=\mathtt{e1}$,
$\mathtt{04}^{-1}=\mathtt{91}$, and
$\mathtt{03}^{-1}=\mathtt{be}$. Proposition~\ref{prop:multiplier-inversion}
therefore explains the pairing in Table~\ref{tab:kuz-inner-selected}.
It does not, however, explain why the common values $64$, $33$, and
$21$ are so large.

This is the structural question considered next. We show that these
values are closely connected to the hidden multiplicative structure of
the Kuznyechik S-box identified by Perrin and Udovenko, together with
the difference between their field representation and the field used
by Kuznyechik.

\section{Origin of the exceptional Kuznyechik multipliers}
\label{sec:kuz-structure}

The previous section shows that the Kuznyechik S-box has unusually large
extended differential entries for a small number of multipliers. We now
explain the origin of these values. The explanation comes from the known
algebraic structure of the S-box, but also from an important mismatch:
the field representation in which this structure appears is not the field
representation used by Kuznyechik.

\subsection{Hidden multiplicative structure of the S-box}
\label{subsec:hidden-structure}

The algebraic structure of the Kuznyechik/Streebog S-box has been studied
in detail by Biryukov, Perrin and Udovenko~\cite{BPU16}. Perrin and
Udovenko later gave a pseudo-exponential representation and identified
strong relations between XOR translations at the S-box output and
multiplication at its input~\cite{PerrinUdovenko2016}.

Their representation uses the field
$\mathbb T=\mathbb F_2[X]/(X^8+X^4+X^3+X^2+1)$.
Kuznyechik, on the other hand, performs byte multiplication in
$\mathbb K=\mathbb F_2[X]/(X^8+X^7+X^6+X+1)$.
We identify both fields with the same vector space $\mathbb F_2^8$
through the polynomial basis. Thus the same byte represents the same
binary vector in both descriptions, but multiplication by that byte can
give two different linear maps.
For a byte $z$, let $M_z^{\mathbb T}$ and $M_z^{\mathbb K}$ denote
multiplication by $z$ in $\mathbb T$ and $\mathbb K$, respectively.

Perrin and Udovenko identified the four output constants
$\mathtt{12}$, $\mathtt{26}$, $\mathtt{24}$, and $\mathtt{30}$,
associated with the hidden-field multipliers
$\mathtt{02}$, $\mathtt{04}$, $\mathtt{10}$, and $\mathtt{1d}$,
respectively~\cite{PerrinUdovenko2016}. For each of these constants,
240 of the 256 S-box values satisfy one of two inverse multiplicative
relations. The two branches have equal size.

\begin{lemma}
\label{lem:perrin-120}
Let $(\lambda,d)$ be one of
$
(\mathtt{02},\mathtt{12}),\quad
(\mathtt{04},\mathtt{26}),\quad
(\mathtt{10},\mathtt{24}),\quad
(\mathtt{1d},\mathtt{30}).
$
Then
$\#\{x:S(M_\lambda^{\mathbb T}x)\oplus S(x)=d\}=120$.
The same count holds for $\lambda^{-1}$.
\end{lemma}

\begin{proof}
Perrin and Udovenko show that, for 240 values of $y$, one of the
relations
$\smash{S^{-1}(y\oplus d)=\lambda S^{-1}(y)}$ or
$\smash{S^{-1}(y\oplus d)=\lambda^{-1}S^{-1}(y)}$
holds in $\mathbb T$~\cite{PerrinUdovenko2016}.
Putting $x=S^{-1}(y)$ gives the two corresponding relations
$S(\lambda x)\oplus S(x)=d$ and
$S(\lambda^{-1}x)\oplus S(x)=d$.

The two solution sets have the same size because
$x\mapsto\lambda x$ maps one branch bijectively to the other.
They are also disjoint. Indeed, a common nonzero solution would imply
$\lambda x=\lambda^{-1}x$ and hence $\lambda^2=1$. In characteristic
two this forces $\lambda=1$, which does not hold for the four
multipliers above. The two branches therefore partition the 240
solutions into two sets of size 120.
\end{proof}

These 120-point transitions are expressed in $\mathbb T$. They do not
immediately give 120-point extended differentials in Kuznyechik,
because the same multiplier acts differently in $\mathbb K$. The
difference between these two multiplication maps is the key to the
observed values.

\subsection{Transfer between the two field representations}
\label{subsec:cross-field-transfer}

Let $M_h$ denote a multiplication map in the hidden representation and
let $M_k$ denote the multiplication map used in the Kuznyechik
representation. Define their disagreement map by
$D=M_k\oplus M_h$.

Suppose that a set of inputs satisfies
$F(M_hx)\oplus F(x)=b$. For any such input,
choosing the offset $a=Dx$ gives
$M_kx\oplus a=M_hx$. Hence the same input also satisfies
$F(M_kx\oplus a)\oplus F(x)=b$.
The hidden transition is therefore distributed among the offsets in
the image of $D$.

This observation gives the following transfer result.

\begin{theorem}
\label{thm:cross-field-transfer}
Let $F:V\rightarrow V$ be a permutation of a finite binary vector
space. Let $M_h$ and $M_k$ be invertible $\mathbb F_2$-linear maps and
put $D=M_k\oplus M_h$. For a fixed $b$, let
$\mathcal H_b=\{x:F(M_hx)\oplus F(x)=b\}$ and
$h=|\mathcal H_b|$. If $r=\operatorname{rank}D$, then
\[
\max_a
\#\{x:F(M_kx\oplus a)\oplus F(x)=b\}
\ge
\left\lceil\frac{h}{2^r}\right\rceil.
\]
\end{theorem}

\begin{proof}
For every $x\in\mathcal H_b$, set $a=Dx$. Then
$M_kx\oplus a=M_hx$, and therefore
$F(M_kx\oplus a)\oplus F(x)=b$.
The image of $D$ contains $2^r$ elements. Hence the $h$ hidden
solutions are assigned to at most $2^r$ different offsets. At least one
offset must therefore contain
$\lceil h/2^r\rceil$ of them.
\end{proof}

Theorem~\ref{thm:cross-field-transfer} reduces the problem to the rank
of the disagreement map. In the present setting this rank can be
determined exactly.

\subsection{Rank of the disagreement map}
\label{subsec:disagreement-rank}

The following result applies whenever the two field representations use
the same polynomial coordinates.

\begin{theorem}
\label{thm:degree-rank}
Let $P,Q\in\mathbb F_2[X]$ be distinct monic irreducible polynomials of
degree $n$. Identify $\mathbb F_2[X]/(P)$ and
$\mathbb F_2[X]/(Q)$ with $\mathbb F_2^n$ using the polynomial basis.
Let $c(X)\ne0$ have degree $d<n$. If $M_c^P$ and $M_c^Q$ denote
multiplication by $c(X)$ modulo $P$ and $Q$, respectively, then
\[
\operatorname{rank}_{\mathbb F_2}
\bigl(M_c^P\oplus M_c^Q\bigr)=d.
\]
\end{theorem}

\begin{proof}
If $\deg x\le n-d-1$, then $\deg(cx)<n$. No modular reduction is
needed in either field. Hence $M_c^Px=M_c^Qx$, so all polynomials of
degree at most $n-d-1$ belong to the kernel.

Conversely, suppose that $M_c^Px=M_c^Qx$. Let $r$ be their common
polynomial representative of degree less than $n$. There exist
polynomials $u$ and $v$ such that
$cx=r+uP=r+vQ$. Hence $uP=vQ$.

Since $P$ and $Q$ are distinct irreducible polynomials of the same
degree, they are coprime. Thus $P\mid v$. If $v\ne0$, division of
$cx$ by the degree-$n$ polynomial $Q$ gives
$\deg v\le d-1<n$, which is impossible for a nonzero multiple of $P$.
Therefore $v=0$, and similarly $u=0$. It follows that $\deg(cx)<n$.
For $x\ne0$, this gives $\deg x\le n-d-1$.

The kernel therefore has dimension $n-d$, and the result follows from rank--nullity.
\end{proof}

For the two field polynomials used here, the degrees of the byte
representatives $\mathtt{02}$, $\mathtt{04}$, $\mathtt{10}$, and
$\mathtt{1d}$ are $1$, $2$, $4$, and $4$. Their disagreement ranks are
therefore $1$, $2$, $4$, and $4$, respectively.

Combining this observation with the 120-point transitions of
Lemma~\ref{lem:perrin-120} already gives the lower bounds
$60$, $30$, $8$, and $8$. No enumeration of the complete Kuznyechik extended differential table is needed to obtain these bounds.

\subsection{Explanation of the exceptional classes}
\label{subsec:exceptional-classes}

Table~\ref{tab:cross-field-certificates} compares the structural bounds
with the exact extended differential uniformities of the Kuznyechik
S-box.
\begin{table}[t]
\centering
\footnotesize
\caption{Structural certificates obtained by transferring hidden-field
multiplicative transitions to the Kuznyechik field representation.}
\label{tab:cross-field-certificates}
\begin{tabular}{c c r r r r}
\toprule
$c$ & $b$ & hidden count $h$ & rank $r$
    & lower bound & $\delta_c(S)$ \\
\midrule
\texttt{02} & \texttt{12} & 120 & 1 & 60 & 64 \\
\texttt{04} & \texttt{26} & 120 & 2 & 30 & 33 \\
\texttt{10} & \texttt{24} & 120 & 4 &  8 & 11 \\
\texttt{1d} & \texttt{30} & 120 & 4 &  8 & 12 \\
\texttt{03}$^\dagger$ & \texttt{a2} & 32 & 1 & 16 & 21 \\
\bottomrule
\end{tabular}
\end{table}
The first four rows come directly from the multiplicative relations
reported by Perrin and Udovenko~\cite{PerrinUdovenko2016}. The
$\mathtt{03}$ row has a different status. Our exhaustive computation in
the hidden representation gives a 32-point transition at output
difference $\mathtt{a2}$. This additional transition is not one of the
four 120-point relations reported in~\cite{PerrinUdovenko2016}; it is
marked by $\dagger$ for this reason. The corresponding verification data
are included in the reproducibility material~\cite{github}.

The rank-one case $c=\mathtt{02}$ is particularly transparent. Its
120 hidden solutions are divided between only two offsets. The two
parts have sizes $64$ and $56$, at offsets $\mathtt{de}$ and
$\mathtt{00}$, respectively. Hence
$N_{02}(\mathtt{de},\mathtt{12})=64$ and
$N_{02}(\mathtt{00},\mathtt{12})=56$.

For $c=\mathtt{04}$, the disagreement rank is two. The corresponding
hidden branch is divided as $32+32+32+24$ among the four offsets
$\mathtt{00}$, $\mathtt{7f}$, $\mathtt{de}$, and $\mathtt{a1}$.
The complete Kuznyechik table has value $33$ at two of these positions,
because one additional solution enters from outside this hidden branch.

For $c=\mathtt{03}$, the 32-point hidden transition is divided into
parts of sizes $21$ and $11$. This directly produces the value $21$
observed in the Kuznyechik extended differential table.

Finally, Proposition~\ref{prop:multiplier-inversion} transfers the same
uniformities to the inverse multipliers. We therefore obtain the three
classes
$\mathtt{02}/\mathtt{e1}$,
$\mathtt{04}/\mathtt{91}$, and
$\mathtt{03}/\mathtt{be}$
with exact uniformities $64$, $33$, and $21$.

The mechanism is not confined to these selected constants. For each nonclassical Kuznyechik multiplier $c$, define
\[
B(c)=\max_{\lambda\in\mathbb F_{256}^*,\,b}
\left\lceil\frac{H_\lambda(b)}
{2^{\operatorname{rank}(M_c^{\mathbb K}\oplus M_\lambda^{\mathbb T})}}\right\rceil,
\]
where
$H_\lambda(b)=\#\{x:S(M_\lambda^{\mathbb T}x)\oplus S(x)=b\}$.
Theorem~\ref{thm:cross-field-transfer} gives $\delta_c\ge B(c)$. An exhaustive
computation over all $254$ nonclassical multipliers gives
$\operatorname{corr}_{\rm Pearson}(B(c),\delta_c)=0.9788451787$.
We use this value only as a descriptive whole-spectrum check: it shows that
the field-representation mismatch explains structure beyond the few
exceptional cells singled out above. The full computation is included in
the reproducibility material~\cite{github}.

The exceptional values are therefore not isolated numerical features
of the extended differential table. They arise from strong multiplicative transitions
in the hidden S-box representation together with the low rank of the
disagreement between the two field multiplications.

This explains the large local transitions. The next question is
different. In the actual cipher, the input to the first S-box is
shifted by the whitening key. We must therefore understand the complete
extended differential table, not only its largest entry, and determine how its
distribution changes with the key.

\section{The complete extended differential table and whitening}
\label{sec:extended-differential-table-whitening}

The previous section explains why a few entries of the Kuznyechik extended differential table are unusually large. For a keyed cipher, however,
the largest entry is not the whole story. The key changes the effective
input offset and therefore selects different rows of the table. We now
study this complete distribution and its interaction with whitening.

\subsection{The normalized extended differential table}

For fixed $c\in\mathbb F_q^*$, let $N_c=(N_c(a,b))_{a,b\in\mathbb F_q}$.
The row sums are immediate from the definition. For permutations, the
column sums are equally simple.

\begin{proposition}
\label{prop:doubly-stochastic}
Let $F$ be a permutation of $\mathbb F_q$. Then, for every
$a,b\in\mathbb F_q$,
\[
\sum_{b\in\mathbb F_q}N_c(a,b)=q,
\qquad
\sum_{a\in\mathbb F_q}N_c(a,b)=q.
\]
Hence $W_c=N_c/q$ is doubly stochastic.
\end{proposition}

\begin{proof}
Fix $a$. Each $x$ determines exactly one value
$b=F(cx\oplus a)\oplus F(x)$, which gives the row sum.

For the column sum, fix $b$ and $x$. Since $F$ is a permutation, there
is a unique $y$ such that $F(y)=F(x)\oplus b$. The required offset is
then uniquely determined by $a=y\oplus cx$. Thus every $x$ contributes
to exactly one row of the column indexed by $b$.
\end{proof}

This property allows us to view $W_c$ as a transition matrix. In
particular, a uniformly distributed row index produces a uniformly
distributed output difference.

\subsection{What the transferred fibers already force}
\label{subsec:structural-information}

The cross-field theorem gives more than a large table entry. The same disagreement fibers force concentration in the corresponding output-difference column and therefore force nonzero information in the complete channel.

For integers $h\ge1$ and $1\le m\le h$, write $h=tm+s$ with $0\le s<m$ and define $\Phi_m(h)=s(t+1)^2+(m-s)t^2$.

\begin{theorem}
\label{thm:cross-field-column-energy}
Under the hypotheses of Theorem~\ref{thm:cross-field-transfer}, suppose the state space has size $q_0$ and $F$ is a permutation. Let $f_a=|\mathcal H_b\cap D^{-1}(a)|$, $h=|\mathcal H_b|$, $r=\operatorname{rank}D$, and $m=\min(2^r,h)$. If $N_p(a,b)=\#\{x:F(M_kx\oplus a)\oplus F(x)=b\}$, then
\[
\sum_a N_p(a,b)^2\ge \Phi_m(h)+(q_0-h).
\]
If the exact fiber sizes are known, the stronger bound
\[
\sum_a N_p(a,b)^2\ge \sum_{a:f_a>0}f_a^2+(q_0-h)
\]
holds.
\end{theorem}

\begin{proof}
The hidden solutions occupy at most $m$ nonempty fibers and their sizes sum to $h$. Convexity of $x^2$ shows that their square sum is minimized by the most balanced partition, giving $\Phi_m(h)$. The remaining column mass is $q_0-h$. Placing each remaining point in an unused row increases the square sum by one, whereas placing it in an already occupied row costs at least three. Hence the minimum additional contribution is $q_0-h$. Using the actual fiber sizes instead of their balanced lower bound gives the stronger bound stated in the theorem.
\end{proof}

For $1\le f\le q_0$, put
\[
\psi_{q_0}(f)=\frac{f}{q_0}\log_2 f+\frac{q_0-f}{q_0}\log_2\!\left(\frac{q_0-f}{q_0-1}\right),
\]
with the usual limiting convention at $f=q_0$.

\begin{theorem}
\label{thm:cross-field-information}
Under the same hypotheses, let $A$ be uniform on the offset space and, conditionally on $A=a$, let $B$ have distribution $N_p(a,\cdot)/q_0$. Both marginals are uniform. If $\mathcal E_p=\sum_{a,y}N_p(a,y)^2$, then
\begin{align*}
\mathcal E_p&\ge q_0^2+\sum_a f_a^2-h,\\
\chi^2(P_{A,B}\Vert P_AP_B)&\ge \frac{\sum_a f_a^2-h}{q_0^2},\\
I(A;B)&\ge \frac1{q_0}\sum_{a:f_a>0}\psi_{q_0}(f_a).
\end{align*}
If only $h$ and $r$ are known, writing $h=tm+s$ with $m=\min(2^r,h)$ gives
\begin{align*}
\chi^2(P_{A,B}\Vert P_AP_B)&\ge \frac{\Phi_m(h)-h}{q_0^2},\\
I(A;B)&\ge \frac{s\psi_{q_0}(t+1)+(m-s)\psi_{q_0}(t)}{q_0}.
\end{align*}
\end{theorem}

\begin{proof}
The same counting argument as in Proposition~\ref{prop:doubly-stochastic} shows that $N_p/q_0$ is doubly stochastic. Theorem~\ref{thm:cross-field-column-energy} gives the required excess square mass in the distinguished column. Every other column has total mass $q_0$ and therefore square sum at least $q_0$ by Cauchy--Schwarz. This proves the energy bound, and the Pearson bound follows from $\chi^2(P_{A,B}\Vert P_AP_B)=\mathcal E_p/q_0^2-1$.

For Shannon information, fix a row with $f=f_a\ge1$ forced solutions in the distinguished output column. Thus the distinguished output has probability at least $f/q_0$. For a fixed value $p$ of that probability, convexity of relative entropy shows that divergence from the uniform distribution is minimized by distributing the remaining mass uniformly over the other $q_0-1$ outputs. The resulting minimum is
\[
G(p)=p\log_2(q_0p)+(1-p)\log_2\!\left(\frac{q_0(1-p)}{q_0-1}\right).
\]
For $p\ge1/q_0$,
\[
G'(p)=\log_2\!\left(\frac{p(q_0-1)}{1-p}\right)\ge0,
\]
so the constraint $p\ge f/q_0$ gives
$D(P_{B\mid A=a}\Vert U)\ge G(f/q_0)=\psi_{q_0}(f)$.
Averaging over the $q_0$ rows proves the stated Shannon-information bound.

It remains to justify the rank-only form. On $1\le f<q_0$,
\[
\psi_{q_0}'(f)=\frac1{q_0}\log_2\!\left(\frac{f(q_0-1)}{q_0-f}\right),
\qquad
\psi_{q_0}''(f)=\frac1{q_0\ln2}\left(\frac1f+\frac1{q_0-f}\right)>0.
\]
Hence $\psi_{q_0}$ is nondecreasing and strictly convex on $[1,q_0)$. Moreover $\psi_{q_0}(1)=0$, and for integers $u,v\ge1$ with $u+v\le q_0$,
\[
\psi_{q_0}(u+v)-\psi_{q_0}(u)
=\int_u^{u+v}\psi_{q_0}'(x)\,dx
\ge \int_1^{v+1}\psi_{q_0}'(x)\,dx
=\psi_{q_0}(v+1)\ge\psi_{q_0}(v).
\]
Thus splitting a positive fiber cannot increase the sum of the $\psi_{q_0}(f_a)$. Since there are at most $m=\min(2^r,h)$ nonempty fibers, the minimum is therefore attained with exactly $m$ nonempty fibers; strict convexity then makes their integer sizes as balanced as possible. Writing $h=tm+s$ gives $s$ fibers of size $t+1$ and $m-s$ of size $t$, which proves the final bound.
\end{proof}

For the Kuznyechik certificates, the exact fiber patterns already imply nonzero information before the remaining entries of the Kuznyechik extended differential table are examined. Table~\ref{tab:structural-information-certificates} gives the resulting bounds together with the exact mutual information computed later from the full table.

\begin{table}[t]
\centering
\scriptsize
\caption{Information forced by a single transferred hidden-field transition. The lower bounds use only the disagreement fibers; the final column uses the complete Kuznyechik extended differential table.}
\label{tab:structural-information-certificates}
\begin{tabular}{c c r r r}
\toprule
$c$ & fiber pattern & $\chi^2$ lower bound & $I$ lower bound (bits) & exact $I$ (bits)\\
\midrule
\texttt{02} & $56+64$ & 0.108521 & 0.008553 & 0.887664\\
\texttt{04} & $32+32+32+24$ & 0.053833 & 0.006603 & 0.858480\\
\texttt{10} & 16 nonempty fibers & 0.012817 & 0.003161 & 0.832516\\
\texttt{1d} & 16 nonempty fibers & 0.012817 & 0.003153 & 0.818254\\
\texttt{03}$^{\dagger}$ & $11+21$ & 0.008087 & 0.001350 & 0.816620\\
\bottomrule
\end{tabular}
\end{table}

\subsection{A second-moment measure of the complete table}
\label{subsec:inner-energy}

The uniformity $\delta_c(F)$ records only the largest entry of $N_c$.
To measure the concentration of the complete table, define
\begin{equation}
\mathcal E_c(F)
=
\sum_{a,b\in\mathbb F_q}N_c(a,b)^2.
\label{eq:inner-energy}
\end{equation}

For $c=1$, this is the usual sum of squares of the DDT\@. Such
second-moment quantities are well studied for ordinary differential
tables~\cite{Cheng2022}; complete-table viewpoints also occur in the
$c$-differential-linear setting~\cite{EddahmaniMesnager2024}. For general $c$, the energy of the extended differential table has a
particularly simple relation with the ordinary DDT\@.

\begin{theorem}
\label{thm:energy-ddt-correlation}
For every $c\in\mathbb F_q^*$,
\[
\mathcal E_c(F)
=
\sum_{\alpha,\beta\in\mathbb F_q}
D_F(\alpha,\beta)D_F(c\alpha,\beta).
\]
\end{theorem}

\begin{proof}
Expanding~\eqref{eq:inner-energy} counts pairs $x,y$ satisfying
\[
F(cx\oplus a)\oplus F(x)
=
F(cy\oplus a)\oplus F(y).
\]
Put $\alpha=x\oplus y$ and $u=cx\oplus a$. Then
$cy\oplus a=u\oplus c\alpha$, and the equality becomes
\[
F(u\oplus c\alpha)\oplus F(u)
=
F(x\oplus\alpha)\oplus F(x).
\]
If their common value is $\beta$, the number of choices for $x$ is
$D_F(\alpha,\beta)$, while the number of choices for $u$ is
$D_F(c\alpha,\beta)$. Summing over $\alpha$ and $\beta$ proves the
claim.
\end{proof}

Thus $\mathcal E_c(F)$ is the scalar product of the ordinary DDT with
the table obtained by the row permutation $\alpha\mapsto c\alpha$.
The quantity $\delta_c(F)$ and the energy therefore describe different
features. The first detects the strongest transition. The second measures
how the complete extended differential distribution is concentrated.

The row-correlation formula yields several useful identities that clarify what the complete-table energy measures.

\begin{proposition}
\label{prop:energy-defect}
For every $c\in\mathbb F_q^*$,
\[
\mathcal E_1(F)-\mathcal E_c(F)=\frac12\sum_{\alpha,\beta}\bigl(D_F(\alpha,\beta)-D_F(c\alpha,\beta)\bigr)^2.
\]
Hence $\mathcal E_c(F)\le\mathcal E_1(F)$, with equality exactly when $D_F(\alpha,\beta)=D_F(c\alpha,\beta)$ for all $\alpha,\beta$.
\end{proposition}

\begin{proof}
Expanding the square gives two copies of $\mathcal E_1(F)$ because $\alpha\mapsto c\alpha$ permutes the DDT rows, and the cross term is $2\mathcal E_c(F)$ by Theorem~\ref{thm:energy-ddt-correlation}.
\end{proof}

Thus $2(\mathcal E_1-\mathcal E_c)$ is the squared Euclidean distance between the ordinary DDT and its row-permuted copy under $\alpha\mapsto c\alpha$.

\begin{proposition}
\label{prop:inner-outer-energy-duality}
Let $\mathcal E_c^{\rm out}(F)=\sum_{a,b}{}_c\Delta_F(a,b)^2$. Then
\[
\mathcal E_c^{\rm out}(F)=\sum_{\alpha,\beta}D_F(\alpha,\beta)D_F(\alpha,c\beta).
\]
If $F$ is a permutation, then
\[
\mathcal E_c(F)=\mathcal E_c^{\rm out}(F^{-1}).
\]
\end{proposition}

\begin{proof}
Expanding the square of the outer table and putting $\alpha=x\oplus y$ and $\beta=F(x)\oplus F(y)$ gives the first identity. For a permutation, $D_{F^{-1}}(\alpha,\beta)=D_F(\beta,\alpha)$; applying the first identity to $F^{-1}$ and interchanging $\alpha,\beta$ gives the second.
\end{proof}

\begin{theorem}
\label{thm:energy-lower}
For every $c\in\mathbb F_q^*$, $\mathcal E_c(F)\ge q^2$, with equality if and only if $N_c(a,b)=1$ for all $a,b\in\mathbb F_q$.
\end{theorem}

\begin{proof}
Proposition~\ref{prop:doubly-stochastic} gives $\sum_{a,b}N_c(a,b)=q^2$. Since there are $q^2$ entries, Cauchy--Schwarz yields
\[
q^4=\left(\sum_{a,b}N_c(a,b)\right)^2\le q^2\sum_{a,b}N_c(a,b)^2=q^2\mathcal E_c(F).
\]
Equality holds exactly when all entries are equal to their common average, which is one.
\end{proof}

For $\alpha,v\in\mathbb F_q$, define
\[
\operatorname{AC}_F(\alpha,v)=\sum_{x\in\mathbb F_q}(-1)^{\operatorname{Tr}(v(F(x\oplus\alpha)\oplus F(x)))},
\]
where $\operatorname{Tr}:\mathbb F_q\to\mathbb F_2$ is the absolute trace.

\begin{theorem}
\label{thm:energy-autocorr}
For every $c\in\mathbb F_q^*$,
\[
\mathcal E_c(F)=\frac1q\sum_{\alpha,v\in\mathbb F_q}\operatorname{AC}_F(\alpha,v)\operatorname{AC}_F(c\alpha,v).
\]
\end{theorem}

\begin{proof}
For fixed $\alpha$, the character transform of $\beta\mapsto D_F(\alpha,\beta)$ is $\operatorname{AC}_F(\alpha,v)$. Parseval's identity gives the corresponding cross-correlation identity for each $\alpha$; summing over $\alpha$ and applying Theorem~\ref{thm:energy-ddt-correlation} proves the result.
\end{proof}

The ordinary case $c=1$ belongs to the established DDT sum-of-squares/autocorrelation theory~\cite{Cheng2022}; the formula above is the cross-direction form induced by the multiplier.

\begin{theorem}
\label{thm:energy-conservation}
For every permutation $F$ of $\mathbb F_q$,
\[
\sum_{c\in\mathbb F_q^*}\mathcal E_c(F)=2(q-1)q^2.
\]
Equivalently, for $\eta_c(F)=\mathcal E_c(F)/q^2-1$,
\[
\sum_{c\in\mathbb F_q^*}\eta_c(F)=q-1.
\]
\end{theorem}

\begin{proof}
From Theorem~\ref{thm:energy-ddt-correlation}, $\sum_{c\ne0}\mathcal E_c(F)=\sum_{\alpha,\beta}D_F(\alpha,\beta)\sum_{c\ne0}D_F(c\alpha,\beta)$. The term $\alpha=0$ contributes $(q-1)q^2$. For $\alpha\ne0$, multiplication by $c$ runs through every nonzero input difference. Bijectivity of $F$ gives $\sum_{\gamma\ne0}D_F(\gamma,\beta)=q$ for every $\beta\ne0$ and zero for $\beta=0$, so the nonzero-$\alpha$ contribution is another $(q-1)q^2$. This proves the first identity; the equivalent normalized form follows immediately.
\end{proof}

The conservation law shows that unusually concentrated multiplier classes must be balanced by others. It therefore supports studying the complete spectrum rather than ranking multipliers only by $\delta_c(F)$.

We use the normalized quantity
\[
\eta_c(F)=\frac{\mathcal E_c(F)}{q^2}-1
\]
below.

\subsection{Effect of whitening}
\label{subsec:whitening}

Consider one byte before the first S-box. Let the two plaintext values
satisfy $x'=cx\oplus a$, and let the same whitening-key byte $K$ be
added to both. If $X=x\oplus K$ and $X'=x'\oplus K$, then
\begin{equation}
X'=cX\oplus A,
\qquad
A=a\oplus(1\oplus c)K.
\label{eq:effective-row}
\end{equation}

Equation~\eqref{eq:effective-row} is the main effect of whitening on an
extended differential. For $c\ne1$, the map
$K\mapsto a\oplus(1\oplus c)K$ is a permutation of $\mathbb F_q$.
Hence different key values select different rows of the same extended differential table.

This behavior is also consistent with the affine viewpoint of
commutative cryptanalysis~\cite{Baudrin23,BaudrinGeneral2025}. If
$A_{c,a}(x)=cx\oplus a$ and $T_K(x)=x\oplus K$, then conjugating
$A_{c,a}$ by $T_K$ changes only its affine offset. The multiplier $c$
is unchanged. For our purpose, the important consequence is the
key-dependent row in~\eqref{eq:effective-row}.

For a fixed output difference $b$, define
$p_K=N_c(A,b)/q$, where $A$ is given by~\eqref{eq:effective-row}.
Since $K\mapsto A$ is bijective, the probabilities obtained when $K$
runs through $\mathbb F_q$ are exactly the row values
$N_c(\alpha,b)/q$, in a different order. Whitening therefore does not
remove the structure of the extended differential table. It turns that structure into
key-dependent transition probabilities.

\subsection{Information about a whitening byte}
\label{subsec:whitening-information}

The previous observation gives an exact information-theoretic
interpretation of the extended differential-table energy. Let $K$ be uniform on
$\mathbb F_q$. For a fixed external offset $a$ and $c\ne1$, let $B$
denote the first-S-box output difference. Then
\[
\Pr[B=b\mid K=k]
=
\frac{N_c(a\oplus(1\oplus c)k,b)}{q}.
\]

\begin{theorem}
\label{thm:whitening-information}
Under the experiment above, both $K$ and $B$ are uniform and
\[
\chi^2(P_{K,B}\Vert P_KP_B)
=
\frac{\mathcal E_c(F)}{q^2}-1.
\]
Moreover,
\[
I(K;B)
=
\frac{1}{q^2}
\sum_{\alpha,b\in\mathbb F_q}
N_c(\alpha,b)\log_2 N_c(\alpha,b),
\]
where $0\log_2 0$ is interpreted as zero.
\end{theorem}

\begin{proof}
Because $c\ne1$, the map
$k\mapsto\alpha=a\oplus(1\oplus c)k$ is a permutation of
$\mathbb F_q$. Hence
$P_{K,B}(k,b)=N_c(\alpha,b)/q^2$.

The key $K$ is uniform by construction. The output $B$ is also uniform
by Proposition~\ref{prop:doubly-stochastic}. Thus
$P_K(k)P_B(b)=1/q^2$. Substitution in the Pearson divergence gives
\[
\chi^2(P_{K,B}\Vert P_KP_B)
=
\frac{1}{q^2}
\sum_{\alpha,b}(N_c(\alpha,b)-1)^2
=
\frac{\mathcal E_c(F)}{q^2}-1.
\]
The mutual-information formula follows from the same joint
distribution.
\end{proof}

The importance of Theorem~\ref{thm:whitening-information} is that
$\mathcal E_c(F)$ now has a direct cryptanalytic meaning. It is not only
a measure of S-box-table concentration. After whitening, it measures
the Pearson information carried by the first-S-box output difference
about the whitening-key byte.

For the principal Kuznyechik multiplier classes, the exact values are
shown in Table~\ref{tab:kuz-energy-information}.

\begin{table}[t]
\centering
\footnotesize
\caption{Selected complete-table quantities for the Kuznyechik S-box.
Here $\eta_c=\mathcal E_c/2^{16}-1$.}
\label{tab:kuz-energy-information}
\begin{tabular}{c r r r r}
\toprule
class & $\delta_c$ & $\mathcal E_c$ & $\eta_c$
      & $I(K;B)$ (bits) \\
\midrule
\texttt{02}/\texttt{e1} & 64 & 145324 & 1.217468 & 0.887664 \\
\texttt{04}/\texttt{91} & 33 & 138108 & 1.107361 & 0.858480 \\
\texttt{03}/\texttt{be} & 21 & 130792 & 0.995728 & 0.816620 \\
\bottomrule
\end{tabular}
\end{table}

These values also show why the maximum entry alone is insufficient.
The three classes have different maximum transition probabilities, but
each produces substantial dependence between the whitening byte and
the first-S-box output difference.

\subsection{Parallel S-boxes}
\label{subsec:parallel-inner}

The bytewise analysis extends directly to a parallel S-box layer. Let
$F^{\otimes s}$ act independently on $s$ coordinates, and let
$\mathbf c=(c_1,\ldots,c_s)$ act componentwise. Then the vector extended differential
table factors coordinatewise.

\begin{theorem}
\label{thm:inner-tensor}
For $\mathbf a,\mathbf b\in\mathbb F_q^s$,
\[
N_{\mathbf c}(\mathbf a,\mathbf b)
=
\prod_{i=1}^{s}N_{c_i}(a_i,b_i).
\]
Equivalently,
$N_{\mathbf c}=N_{c_1}\otimes\cdots\otimes N_{c_s}$.
Consequently, the corresponding table energies multiply.
For independent active whitening bytes, the mutual information adds
over the active coordinates.
\end{theorem}

\begin{proof}
The vector equation separates into $s$ independent scalar equations.
Hence its solution set is the Cartesian product of the corresponding
scalar solution sets, which gives the stated factorization.
The matrix and energy statements follow immediately. The whitening
channel is likewise a product of the scalar channels of
Theorem~\ref{thm:whitening-information}, so mutual information is
additive.
\end{proof}

This factorization allows the multiplier to be localized. One may use a
nontrivial multiplier on a selected byte while leaving the remaining
coordinates unchanged. It also gives a simple count of keys for which
selected effective rows occur.

\begin{proposition}
\label{prop:key-class-count}
Let $J$ be the set of coordinates with $c_i\ne1$. For each $i\in J$,
let $\mathcal A_i\subseteq\mathbb F_q$ be a set of desired effective
rows. If the whitening bytes are independent and uniform, then the
fraction of keys satisfying
$a_i\oplus(1\oplus c_i)K_i\in\mathcal A_i$ for every $i\in J$ is
\[
\prod_{i\in J}\frac{|\mathcal A_i|}{q}.
\]
\end{proposition}

\begin{proof}
For each $i\in J$, the map
$K_i\mapsto a_i\oplus(1\oplus c_i)K_i$ is a permutation of
$\mathbb F_q$. Exactly $|\mathcal A_i|$ of the $q$ key-byte values are
therefore admissible. The constraints concern distinct coordinates, so
their fractions multiply.
\end{proof}

\paragraph{A structured $2^{-7}$ key class.}
Consider one active multiplier byte with $c=\mathtt{02}$ and external
offset $a=0$. The rank-one certificate of
Section~\ref{sec:kuz-structure} gives the two effective rows
$\mathtt{00}$ and $\mathtt{de}$. Since
$1\oplus\mathtt{02}=\mathtt{03}$ in the Kuznyechik field, these rows
correspond to the two whitening-byte values
$\mathtt{00}$ and $\mathtt{4a}$. Thus exactly two of the 256 values of
that key byte are selected. The resulting master-key class has fraction
$2/256=2^{-7}$, or $2^{249}$ keys out of the $2^{256}$ possible master
keys.

The exact two-round channel derived in Theorem~\ref{thm:exact-q2-channel} makes this statement explicit. For input byte $j=4$, output byte $i=6$, diffusion coefficient $\ell_{ij}=\mathtt{06}$, and observed second-S-box difference $\gamma=\mathtt{9c}$, the two effective rows satisfy
\[
Q_{\mathtt{02},\mathtt{06},\mathtt{9c}}(\mathtt{00})=\frac{794}{65536},
\qquad
Q_{\mathtt{02},\mathtt{06},\mathtt{9c}}(\mathtt{de})=\frac{768}{65536}=\frac{3}{256}.
\]
Thus every key in this $2^{-7}$ class has this two-round event with probability at least $3/256$. This is a statement about the reduced two-round setting; no corresponding full-round key-class security claim is made.

\subsection{Exact random-permutation baseline for masked affine pairs}
\label{subsec:exact-baseline}

The statistical experiment later compares masked output classes with a random-permutation reference. For the affine pairing used here, the single-trial ensemble probability can be written exactly.

Let $V$ be an $N$-element vector space over $\mathbb F_2$, let $\pi:V\to W$ be a surjective linear projection with $|W|=M$, and let $P$ be uniformly random over all permutations of $V$.

\begin{lemma}
\label{lem:random-perm-diff}
For distinct fixed inputs $u\ne v$, the difference $P(u)\oplus P(v)$ is uniform on $V\setminus\{0\}$. Consequently,
\begin{align*}
\Pr[\pi(P(u)\oplus P(v))=0]&=\frac{N/M-1}{N-1},\\
\Pr[\pi(P(u)\oplus P(v))=b]&=\frac{N/M}{N-1},\qquad b\ne0.
\end{align*}
\end{lemma}

\begin{proof}
The ordered pair $(P(u),P(v))$ is uniform over the $N(N-1)$ ordered pairs of distinct elements. Every nonzero full difference occurs for exactly $N$ such pairs. The kernel of $\pi$ contains $N/M$ vectors, one of which is zero, while every nonzero coset contains $N/M$ nonzero vectors. Summing the uniform full-difference probability over the relevant coset gives the two formulas.
\end{proof}

\begin{theorem}
\label{thm:exact-masked-baseline}
Let $V=(\mathbb F_q)^s$ and $A_{c,a}(x)=cx\oplus a$, with scalar multiplication applied componentwise. Choose $x$ uniformly from $V$, and let $P$ be a uniformly random permutation of $V$.
\begin{enumerate}
\item If $c\ne1$, then $A_{c,a}$ has exactly one fixed point and
\[
\Pr\!\left[\pi\bigl(P(A_{c,a}(x))\oplus P(x)\bigr)=b\right]=\frac1M
\]
for every $b\in W$.
\item If $c=1$ and $a\ne0$, the two inputs are always distinct and the probabilities are those of Lemma~\ref{lem:random-perm-diff}.
\end{enumerate}
\end{theorem}

\begin{proof}
For $c\ne1$, the equation $cx\oplus a=x$ has the unique solution $x=(1\oplus c)^{-1}a$. Conditioning on whether this fixed point is sampled and applying Lemma~\ref{lem:random-perm-diff} to all other inputs gives probability $1/M$ for every projected output class. The case $c=1$, $a\ne0$ follows directly from the lemma.
\end{proof}

\begin{proposition}
\label{prop:vector-random-baseline}
Let $A_{C,a}(x)=Cx\oplus a$ on $V=(\mathbb F_q)^s$, where $N=|V|=q^s$, $C=\operatorname{diag}(c_0,\ldots,c_{s-1})$, and $J=\{i:c_i\ne1\}$. If $a_i=0$ for every $i\notin J$, then $A_{C,a}$ has $q^{s-|J|}$ fixed points; otherwise it has none. If $f$ is the resulting fixed-point fraction, then
\begin{align*}
\Pr[\pi(P(A_{C,a}(x))\oplus P(x))=0]&=f+(1-f)\frac{N/M-1}{N-1},\\
\Pr[\pi(P(A_{C,a}(x))\oplus P(x))=b]&=(1-f)\frac{N/M}{N-1},\qquad b\ne0.
\end{align*}
\end{proposition}

\begin{proof}
The fixed-point equation is coordinatewise: $(1\oplus c_i)x_i=a_i$. If $c_i\ne1$ it has one solution; if $c_i=1$ it requires $a_i=0$ and then leaves $x_i$ free. Conditioning on whether $x$ is fixed and applying Lemma~\ref{lem:random-perm-diff} gives the stated probabilities.
\end{proof}

For the one-active-byte vector multiplier used above, $|J|=1$ and hence $f=1/q$. With a one-byte output projection ($M=q=256$), a nonzero class has ideal-permutation probability $\left(1-1/256\right)2^{120}/(2^{128}-1)=0.0038909912\ldots$. For the scalar-$c$ one-byte-to-one-byte experiment used later, the sampled input-offset class is uniform over $q-1$ nonzero byte values; Theorem~\ref{thm:exact-masked-baseline} therefore gives the joint-cell reference $P_{\rm exp}=1/((q-1)q)=1/((2^8-1)2^8)$ for $c\ne1$.

These statements concern an ensemble in which the permutation itself is random. The experiment in Section~\ref{sec-cddt-model} fixes one realized Kuznyechik permutation and randomizes only plaintexts and affine offsets. The distinction remains essential when interpreting its
statistical evidence.

The discussion so far concerns the information available immediately
after the first nonlinear layer. The next question is whether this
key-dependent distribution remains distinguishable after another
S-box layer. For two-round standard Kuznyechik, this question can be
answered exactly.

\section{Exact two-round key recovery against standard Kuznyechik}
\label{sec:standard-2r-recovery}

The previous section shows that whitening selects a key-dependent row of
the extended differential table. We now study whether this dependence
remains visible after a second S-box layer. For two-round standard
Kuznyechik, the resulting distribution can be derived exactly.

Let $S$ denote the bytewise Kuznyechik S-box and let $L$ denote its
$16$-byte linear layer. In this section, $K_1,K_2,K_3$ denote the first
three $128$-bit round keys. The reduced two-round encryption is
\begin{equation}
E_K^{(2)}(P)
=
L\!\left(
S\!\left(
L(S(P\oplus K_1))\oplus K_2
\right)
\right)\oplus K_3.
\label{eq:standard-2r}
\end{equation}
Under the Kuznyechik key schedule~\cite{RFC7801}, $K_1$ and $K_2$ are the two
$128$-bit halves of the master key. Recovering them therefore recovers
the complete $256$-bit key.

All probabilities in this section are derived for a fixed key. They do
not rely on a Markov or stochastic-independence assumption between the
two rounds. This distinction is important because differential
characteristics can exhibit key dependencies~\cite{PeyrinTan2022}, and
fixed-key differential analysis requires keeping such dependencies explicit~\cite{BeyneRijmen2022}.

\subsection{Stage 1: recovering the first master-key half}
\label{subsec:stage1}

Fix an input byte position $j$, a nontrivial multiplier
$c\in\mathbb F_{256}$, and an external offset $a$. For a uniformly
random plaintext $P$, form $P'$ by
$P'_j=cP_j\oplus a$ and $P'_t=P_t$ for $t\ne j$.

Let $X=P\oplus K_1$. At byte $j$,
\[
X'_j=cX_j\oplus\alpha,
\qquad
\alpha=a\oplus(1\oplus c)K_{1,j}.
\]
Hence the first S-box output difference is supported only at byte $j$,
with
\[
\Pr[B_j=b\mid K_{1,j}]
=
\frac{N_c(\alpha,b)}{q},
\qquad q=256.
\]

Write $L=(\ell_{it})_{0\le i,t<16}$ over $\mathbb F_{256}$. Since the
Kuznyechik linear layer has branch number $17$~\cite{AY15}, every
coefficient $\ell_{it}$ is nonzero. Otherwise, a one-active-byte input
in column $t$ would give at most $15$ active output bytes.

Let $C=E_K^{(2)}(P)$ and $C'=E_K^{(2)}(P')$, and define the observable
$\Gamma=L^{-1}(C\oplus C')$.

\begin{theorem}
\label{thm:exact-q2-channel}
For every output byte $i$ and every $\gamma\in\mathbb F_{256}$,
\[
\Pr[\Gamma_i=\gamma\mid K_{1,j}=k]
=
Q_{c,\ell_{ij},\gamma}
\bigl(a\oplus(1\oplus c)k\bigr),
\]
where
\[
Q_{c,\ell,\gamma}(\alpha)
=
\frac{1}{q^2}
\sum_{b\in\mathbb F_q}
N_c(\alpha,b)D_S(\ell b,\gamma).
\]
The formula holds for every fixed choice of $K_1,K_2,K_3$.
\end{theorem}

\begin{proof}
Condition on the first S-box output difference $B_j=b$. After the first
linear layer, the difference entering the second S-box at byte $i$ is
exactly $\ell_{ij}b$. The addition of $K_2$ cancels in the XOR
difference.

It remains to show that the base input to this second S-box is uniform
after conditioning on $B_j=b$. Choose any byte position $t\ne j$.
The byte $P_t$ is independent of the event $B_j=b$ and is uniform.
Since the S-box is bijective and $\ell_{it}\ne0$, the term
$\ell_{it}S(P_t\oplus K_{1,t})$ is uniform. It appears as an independent
summand in the second-S-box input at byte $i$. Hence that input remains
uniform after conditioning on $B_j=b$.

The corresponding second-S-box transition therefore has probability
$D_S(\ell_{ij}b,\gamma)/q$. Multiplying this by
$N_c(\alpha,b)/q$ and summing over $b$ gives the channel formula stated in Theorem~\ref{thm:exact-q2-channel}. Finally, $K_3$ cancels in $C\oplus C'$, and
$L^{-1}$ exposes the second S-box output difference.
\end{proof}

The exact channel has a useful matrix form. Let
$W_c=N_c/q$ and $W_1=D_S/q$. For $\ell\ne0$, let $P_\ell$ be the
permutation matrix corresponding to multiplication by $\ell$.

\begin{theorem}
\label{thm:q2-factorization}
For every $c$ and every nonzero $\ell$,
\[
Q_{c,\ell}=W_cP_\ell W_1.
\]
If $U$ denotes the $q\times q$ uniform channel, then
\[
Q_{c,\ell}-U
=
(W_c-U)P_\ell(W_1-U).
\]
\end{theorem}

\begin{proof}
The $(\alpha,\gamma)$ entry of $W_cP_\ell W_1$ is exactly the defining sum for $Q_{c,\ell}$ in Theorem~\ref{thm:exact-q2-channel}. The centered identity follows because
$W_c$, $P_\ell$, and $W_1$ are doubly stochastic.
\end{proof}

Theorem~\ref{thm:q2-factorization} separates two effects. The first
extended-differential interface creates the key-dependent distribution. The second
S-box attenuates this distribution through the ordinary DDT\@.

\subsection{Information contraction and key identifiability}
\label{subsec:channel-identifiability}

For every permutation $F$ and every $c\ne1$, the zero-output column of
the extended differential table is particularly simple.

\begin{lemma}
\label{lem:zero-column}
For every $a\in\mathbb F_q$ and $c\ne1$,
\[
N_c(a,0)=1.
\]
\end{lemma}

\begin{proof}
Since $F$ is bijective,
$F(cx\oplus a)=F(x)$ is equivalent to
$(1\oplus c)x=a$, which has exactly one solution.
\end{proof}

For the normalized ordinary DDT, write $W_1^*$ for the
$(q-1)\times(q-1)$ block indexed by nonzero differences. Define
\[
\sigma_\perp(S)
=
\sup_{\substack{z\ne0\\ z\mathbf 1=0}}
\frac{\|zW_1^*\|_2}{\|z\|_2}.
\]
This is the contraction factor of the ordinary S-box channel on the
zero-sum subspace.

\begin{theorem}
\label{thm:spectral-contraction}
Let $c\ne1$ and $\ell\ne0$. For a uniform effective row $A$,
\[
\chi^2_{2R}(c,\ell)
\le
\sigma_\perp(S)^2
\left(
\frac{\mathcal E_c}{q^2}-1
\right),
\]
where $\chi^2_{2R}(c,\ell)$ is the Pearson information of the
two-round channel.
\end{theorem}

\begin{proof}
By Lemma~\ref{lem:zero-column}, the zero-output coordinate of
$W_c-U$ vanishes. Multiplication by $P_\ell$ preserves this property.
Every row of $(W_c-U)P_\ell$ therefore lies in the zero-sum subspace
on which $W_1^*$ contracts by at most $\sigma_\perp(S)$.

Applying this bound row by row to the centered identity in Theorem~\ref{thm:q2-factorization} gives
\[
\|Q_{c,\ell}-U\|_F^2
\le
\sigma_\perp(S)^2\|W_c-U\|_F^2.
\]
Theorem~\ref{thm:whitening-information} gives
$\|W_c-U\|_F^2=\mathcal E_c/q^2-1$.
\end{proof}

For the Kuznyechik S-box, numerical singular-value computation gives
$\sigma_\perp(S)\approx0.210276$ and
$\sigma_\perp(S)^2\approx0.044216$.
For the classes
$\mathtt{02}/\mathtt{e1}$,
$\mathtt{04}/\mathtt{91}$, and
$\mathtt{03}/\mathtt{be}$,
the actual two-round ratios over the nonzero coefficients of $L$ are
approximately between $0.0065$ and $0.0071$.

The contraction is therefore strong. However, a reduction in the
magnitude of the information does not imply that two key hypotheses
become identical.

\begin{theorem}
\label{thm:row-identifiability}
Assume that $W_1$ is nonsingular and $\ell\ne0$. Then
\[
Q_{c,\ell}(\alpha,\cdot)
=
Q_{c,\ell}(\alpha',\cdot)
\]
if and only if
\[
W_c(\alpha,\cdot)
=
W_c(\alpha',\cdot).
\]
Moreover,
\[
\operatorname{rank}Q_{c,\ell}
=
\operatorname{rank}W_c.
\]
\end{theorem}

\begin{proof}
By Theorem~\ref{thm:q2-factorization},
\[
Q_{c,\ell}(\alpha,\cdot)
-
Q_{c,\ell}(\alpha',\cdot)
=
\bigl(
W_c(\alpha,\cdot)-W_c(\alpha',\cdot)
\bigr)P_\ell W_1.
\]
Both $P_\ell$ and $W_1$ are invertible, so the right-hand side vanishes
exactly when the original row difference vanishes. The same invertible
right multiplication preserves matrix rank.
\end{proof}

For the Kuznyechik S-box, the nonzero $255\times255$ block of the
ordinary DDT is nonsingular. Exact integer elimination gives
$\det D_S^*\equiv1\pmod 5$. In addition, exhaustive construction of
the extended differential tables shows that the $256$ rows of $N_c$ are pairwise
distinct for every $c\in\mathbb F_{256}^*\setminus\{1\}$. These exact
certificates are included in the reproducibility material~\cite{github}.

\begin{corollary}
\label{cor:key-identifiability}
For every nontrivial $c$ and every nonzero diffusion coefficient
$\ell$, the $256$ whitening-byte hypotheses give $256$ pairwise
distinct two-round distributions.
\end{corollary}

This is an identifiability statement. It does not imply that all
multipliers require the same number of samples. The separation between
the distributions can vary substantially with $c$ and $\ell$.

\subsection{Maximum-likelihood recovery of $K_1$}
\label{subsec:stage1-ml}

For the concrete attack we use $c=\mathtt{04}$. This multiplier has a
direct structural explanation through the cross-field analysis of
Section~\ref{sec:kuz-structure} and also gives strong statistical
separation.

For each input byte $j$, we select the output byte $i$ maximizing the worst pairwise separation among the $256$ key hypotheses. Table~\ref{tab:stage1-channel-design} records the selected channels; the worst separation is used below.

\begin{table}[t]
\centering
\footnotesize
\caption{Exact $c=\mathtt{04}$ Stage-1 channel design.
$D_{\min}=$ $-\ln\rho$, where $\rho$ is the largest Bhattacharyya
coefficient between two distinct key-byte distributions.}
\label{tab:stage1-channel-design}
\begin{tabular}{c c c c @{\qquad} c c c c}
\toprule
$j$ & $i$ & $\ell_{ij}$ & $D_{\min}$
&
$j$ & $i$ & $\ell_{ij}$ & $D_{\min}$ \\
\midrule
0  & 11 & \texttt{bd} & 0.001020865 &
8  & 0  & \texttt{0a} & 0.000996427 \\
1  & 12 & \texttt{e9} & 0.001003868 &
9  & 9  & \texttt{54} & 0.000989626 \\
2  & 4  & \texttt{6b} & 0.000995577 &
10 & 8  & \texttt{2a} & 0.001016197 \\
3  & 9  & \texttt{2f} & 0.000976755 &
11 & 7  & \texttt{b1} & 0.001003688 \\
4  & 11 & \texttt{e9} & 0.001003868 &
12 & 2  & \texttt{9f} & 0.000965553 \\
5  & 5  & \texttt{bc} & 0.000964278 &
13 & 7  & \texttt{2a} & 0.001016197 \\
6  & 13 & \texttt{5a} & 0.000981690 &
14 & 3  & \texttt{e8} & 0.001012424 \\
7  & 13 & \texttt{e0} & 0.000986993 &
15 & 10 & \texttt{bd} & 0.001020865 \\
\bottomrule
\end{tabular}
\end{table}

We use maximum likelihood to identify the key byte. The following bound
gives a rigorous sample complexity.

\begin{theorem}
\label{thm:ml-bound}
Let $\{P_k:k\in\mathcal K\}$ be a finite family of distributions on
an alphabet $\mathcal Y$. Suppose
$Y_1,\ldots,Y_N$ are independent samples from the true distribution
$P_k$. Let
\[
\rho
=
\max_{k\ne k'}
\sum_{y\in\mathcal Y}
\sqrt{P_k(y)P_{k'}(y)}.
\]
Then the maximum-likelihood decoder satisfies
\[
\Pr[\widehat k\ne k]
\le
(|\mathcal K|-1)\rho^N.
\]
For $s$ decoded coordinates,
\[
\Pr[\widehat K\ne K]
\le
s(|\mathcal K|-1)\rho^N.
\]
Hence it is sufficient to take
\[
N
\ge
\left\lceil
\frac{\ln(s(|\mathcal K|-1)/\varepsilon)}
{-\ln\rho}
\right\rceil
\]
to obtain total decoding error at most $\varepsilon$.
\end{theorem}

\begin{proof}
Fix an incorrect hypothesis $k'$. Let
\[
R=
\prod_{r=1}^{N}
\frac{P_{k'}(Y_r)}{P_k(Y_r)}
\]
be its likelihood ratio. If $k'$ has likelihood at least that of the
true key, then $R\ge1$. Applying Markov's inequality to $\sqrt R$ gives
\[
\Pr[R\ge1]
\le
\mathbb E_k[\sqrt R]
=
\left(
\sum_y\sqrt{P_k(y)P_{k'}(y)}
\right)^N
\le\rho^N.
\]
A union bound over the incorrect hypotheses proves the first bound. A second union bound over the $s$ key coordinates gives the second.
\end{proof}

For the selected channels, the worst separation is $D_{\min}^{(1)}=0.000964278$.
Taking $s=16$, $|\mathcal K|=256$, and the Stage-1 failure budget
$\varepsilon_1=0.005$ gives $N_1=14\,117$.

The same $N_1$ base plaintexts are reused across the $16$ byte
experiments. Stage 1 therefore requires $16N_1=225\,872$ chosen pairs
and at most $17N_1=239\,989$ encryption queries.

An exhaustive scan over all nontrivial multipliers was also performed.
The numerically best Stage-1 inverse class is
$\mathtt{4f}/\mathtt{aa}$ and requires $13\,884$ base samples under the
same error budget. The class $\mathtt{04}/\mathtt{91}$ requires
$14\,117$, only about $1.68\%$ more. We use $c=\mathtt{04}$ because it
is directly connected to the structural mechanism of
Section~\ref{sec:kuz-structure}, rather than because it is the
numerically best multiplier in the exhaustive scan. The complete scan is included in the reproducibility material~\cite{github}.

For the RFC test key, Stage 1 recovers the correct first $128$-bit master-key half.

\subsection{Stage $2$: recovering the second master-key half}
\label{subsec:stage2}

Once $K_1$ is known, the first round can be inverted at the plaintext
interface. Let
\[
V=L(S(P\oplus K_1)).
\]
Because both $S$ and $L$ are bijective, any chosen internal state $V$
can be realized using
\begin{equation}
P
=
K_1\oplus S^{-1}(L^{-1}(V)).
\label{eq:stage2-preimage}
\end{equation}

We can therefore place an extended differential relation immediately before the
second S-box.

\begin{theorem}
\label{thm:stage2-channel}
Fix a byte position $j$, $c\ne1$, and $a\in\mathbb F_q$. Choose $V$
uniformly and define $V'$ by
$V'_j=cV_j\oplus a$ and $V'_t=V_t$ for $t\ne j$.
Construct the corresponding plaintext pair using
\eqref{eq:stage2-preimage}. If
\[
B=[L^{-1}(C\oplus C')]_j,
\]
then, for every fixed key byte $K_{2,j}=k$,
\[
\Pr[B=b\mid K_{2,j}=k]
=
\frac{N_c(a\oplus(1\oplus c)k,b)}{q}.
\]
\end{theorem}

\begin{proof}
The second S-box inputs are $U=V\oplus K_2$ and
$U'=V'\oplus K_2$. At byte~$j$,
\[
U'_j
=
cU_j\oplus a\oplus(1\oplus c)K_{2,j}.
\]
Since $V_j$ is uniform, so is $U_j$. The defining extended differential-table count therefore gives the row-channel formula stated in Theorem~\ref{thm:stage2-channel}. The final key addition
cancels, and $L^{-1}$ exposes the second S-box output difference.
\end{proof}

For $c=\mathtt{04}$, all $256$ rows of $N_c$ are distinct. The worst
pairwise Bhattacharyya coefficient is
$\rho_2=0.7074682424211026$, corresponding to
$D_{\min}^{(2)}=0.346063$.
Using the Stage-2 error budget $\varepsilon_2=0.005$,
Theorem~\ref{thm:ml-bound} gives
$N_2=40$
base samples per byte. Thus Stage 2 requires $640$ chosen pairs and
$680$ encryption queries.

For the RFC test key, Stage 2 recovers the correct second $128$-bit master-key half.

As an implementation check, the second stage was also repeated for
$100$ independently generated master keys, and all $100$ were recovered
correctly; the regression output is included in the reproducibility material~\cite{github}.

\subsection{Complete two-round attack}
\label{subsec:complete-2r-attack}

Combining the two stages gives the main cryptanalytic result.

\begin{corollary}
\label{cor:full-two-round-recovery}
For the reduced two-round standard Kuznyechik model
\eqref{eq:standard-2r}, the $c=\mathtt{04}$ attack recovers the complete
$256$-bit master key with failure probability at most $0.01$.
It uses $226\,512$ chosen-plaintext pairs and at most $240\,669$ encryption queries.
\end{corollary}

\begin{proof}
The two stages recover $K_1$ and $K_2$, which are the two halves of the
master key. The pair count is
$16(N_1+N_2)=226\,512$, while reuse of the common base plaintexts gives
at most $17(N_1+N_2)=240\,669$ encryption queries.

The two error events need not be independent. A union bound gives total
failure probability at most
$\varepsilon_1+\varepsilon_2=0.01$.
\end{proof}

After the output histograms have been formed, a direct implementation
evaluates at most $16\cdot256^2=2^{20}$ hypothesis--symbol score
contributions in each stage. The encryption-query cost therefore
dominates the reported attack complexity.

The result is exact for two reduced rounds of standard Kuznyechik.
Extending it to additional rounds would require tracking the joint
nonlinear state distribution together with the dependence on further
round keys. We make no such extrapolation here.

\section{Reduced-round extended-differential characteristic analysis}
\label{sec:part}

The exact two-round attack of Section~\ref{sec:standard-2r-recovery}
uses the complete fixed-key distribution. We now consider a different
question. We ask whether the large extended-differential transitions of the first
S-box layer also improve explicit differential characteristics over a
few rounds.

The nonclassical relation is used only at the first nonlinear layer.
After that layer, the two states differ by an ordinary XOR difference,
as shown in Lemma~\ref{lem:first-interface}. The remaining transitions
are therefore evaluated with the ordinary DDT\@.

The probabilities in this section are characteristic scores under the
usual product model. They should not be confused with the exact
fixed-key distributions derived in the previous section. Dependencies
between differential transitions can invalidate a naive product of local
probabilities; this issue has been studied in the key-dependent and fixed-key
settings~\cite{PeyrinTan2022,BeyneRijmen2022} and is also addressed by
automated tools such as Trail-Estimator and
AutoDiVer~\cite{TrailEstimator2025,AutoDiVer2025}.
Accordingly, the results below are stated only for the explicit search
families considered here.

Kiryukhin computed the exact classical two-round MEDP of Kuznyechik as
$2^{-86.66\ldots}$~\cite{Kiryukhin18,Kiryukhin19} and later developed
bounds for non-minimum-weight two-round LSX differentials~\cite{Kiryukhin2021}. Those results concern a broader optimization
problem. Our purpose here is to compare classical and extended-differential
characteristics under the same restricted search conditions.

\subsection{First-round interface}
\label{subsec:reduced-interface}

For an output difference $b\in\mathbb F_q$, define
$m_c(b)=\max_a N_c(a,b)$. In the searches below the external offset is
restricted to be nonzero, so we use
$m_c^\star(b)
=
\max_{a\in\mathbb F_q^*}N_c(a,b)$.
For $c=1$, this is the usual nonzero-input DDT maximum associated with
the prescribed output difference.

Consider a parallel layer of $s$ S-boxes. Let $J$ be the set of active
positions at the first nonlinear layer. For a prescribed first-layer
output difference $B=(b_1,\ldots,b_s)$, the best first-layer probability
within this model is
\begin{equation}
\Phi^\star_{c,J}(B)
=
q^{-|J|}
\prod_{i\in J}m_c^\star(b_i).
\label{eq:first-layer-score}
\end{equation}

Once $B$ is fixed, the remaining rounds use only ordinary XOR
differences. Let $H_r(B)$ denote the best continuation probability from
the boundary $L(B)$ through the remaining $r-1$ nonlinear layers,
under the same search restrictions.

\begin{theorem}
\label{thm:first-interface}
Under the product model,
\[
P^{(r)}_{c,J}
=
\max_{B\in\Omega}
\Phi^\star_{c,J}(B)H_r(B),
\]
where $\Omega$ is the admissible set of first-layer output differences.
The continuation factor $H_r(B)$ is independent of $c$.

For a fixed boundary $B$, the gain over the classical case is
\[
\log_2
\frac{\Phi^\star_{c,J}(B)}
     {\Phi^\star_{1,J}(B)}
=
\sum_{i\in J}
\left(
\log_2 m_c^\star(b_i)
-
\log_2 m_1^\star(b_i)
\right).
\]
\end{theorem}

\begin{proof}
After the first S-box layer, the pair is related by the ordinary XOR
difference $B$. The linear layer maps it to $L(B)$, and subsequent key
additions cancel. All later nonlinear transitions are therefore scored
with the ordinary DDT and do not introduce another $c$-dependent
factor. Multiplying the first-layer score by the continuation probability gives the first identity in Theorem~\ref{thm:first-interface}. The logarithmic identity
follows directly from~\eqref{eq:first-layer-score}.
\end{proof}

For Kuznyechik, the linear layer has branch number $17$~\cite{AY15}.
Hence a one-active-byte difference after $L$ requires all $16$ bytes
before $L$ to be active. The topology considered first is therefore $16\longrightarrow1$.
This is useful for extended-differential analysis because the first-layer gain can
accumulate across all sixteen S-boxes.

\subsection{Two- and three-round search}
\label{subsec:two-three-search}

We consider the no-pre-whitening reduced core used in the original
characteristic search. The absence of the initial key addition should
be kept in mind; this is distinct from the standard two-round model of
Section~\ref{sec:standard-2r-recovery}.

Let $k\in\{0,\ldots,15\}$ be the position of the unique active byte
after the first linear layer, and let
$\beta\in\mathbb F_{2^8}^*$ be its value. Applying $L^{-1}$ gives the
required first-S-box output difference
$\Delta'=(\delta'_0,\ldots,\delta'_{15})$.
Its first-round score is
\begin{equation}
P_{r1}
=
q^{-16}
\prod_{i=0}^{15}m_c^\star(\delta'_i).
\label{eq:pr1-score}
\end{equation}

Using~\eqref{eq:pr1-score}, the search exhausts all
$16\cdot255$ choices of $(k,\beta)$. In the
second round, the single active byte $\beta$ passes through one ordinary
S-box. For an output $\gamma$, its contribution is
$\log_2 D_S(\beta,\gamma)-8$.

For three rounds, the resulting difference is propagated through the
second linear layer, and the sixteen S-box transitions of the third
round are evaluated with the ordinary DDT\@. Thus the total
three-round score is the sum of the logarithmic contributions of the
three S-box layers.

Table~\ref{tab:multiround} summarizes the best values obtained by the
search.

\begin{table}[t]
\centering
\footnotesize
\caption{Best reduced-round product-model characteristic scores found in
the stated restricted search families. The advantage is measured relative
to the classical case $c=\mathtt{01}$.}
\label{tab:multiround}
\begin{tabular}{c c c c c c}
\toprule
$c$ & 2 rounds & 3 rounds & 4 rounds
& Adv.\ (2R) & Adv.\ (3R) \\
\midrule
\texttt{01} & $2^{-89.2}$ & $2^{-174.3}$ & $2^{-200.8}$ & -- & -- \\
\texttt{02} & $2^{-84.0}$ & $2^{-169.7}$ & $2^{-199.2}$ & $+5.2$ & $+4.6$ \\
\texttt{03} & $2^{-91.2}$ & $2^{-177.4}$ & $2^{-205.1}$ & $-2.0$ & $-3.1$ \\
\texttt{e1} & $2^{-84.0}$ & $2^{-169.7}$ & $2^{-199.2}$ & $+5.2$ & $+4.6$ \\
\bottomrule
\end{tabular}
\end{table}

\paragraph{Two rounds}
The best classical characteristic in this family has
$\log_2 P=-89.15$. It uses $k=4$ and $\beta=\mathtt{0c}$, with
first-layer counts
$(6,8,6,8,8,6,8,8,6,6,6,8,6,6,6,8)$.

For $c=\mathtt{02}$, the best characteristic has
$\log_2 P=-83.98$. It uses $k=2$ and $\beta=\mathtt{91}$, with
first-layer counts
$(4,7,64,6,64,64,5,4,5,6,6,8,5,5,9,4)$.
The inverse multiplier $c=\mathtt{e1}$ gives the same optimum, as
expected from Proposition~\ref{prop:multiplier-inversion}.

The resulting gain is about $5.2$ bits. The three entries equal to
$64$ contribute strongly, but they do not by themselves explain the
full gain. Theorem~\ref{thm:first-interface} shows that all
sixteen bytewise ratios contribute. Moreover, the classical and
nonclassical searches need not choose the same first-round boundary.

\paragraph{Three rounds}
For $c=\mathtt{02}$, the best characteristic found in the restricted
$16\to1\to16$ family has probability approximately $2^{-169.7}$,
compared with $2^{-174.3}$ in the classical search. The gain is therefore
about $4.6$ bits.

One corresponding characteristic has
$k=2$, $\beta=\mathtt{91}$ and $\gamma=\mathtt{f0}$, with
$\log_2 P_{r1}=-77.98, 
\log_2 P_{r2}=-6.00, 
\log_2 P_{r3}=-85.74$.
The first-layer counts are
$(4,7,64,6,64,64,5,4,\allowbreak 5,6,6,8,5,5,9,4)$,
while the third-layer counts are
$(6,8,8,6,4,6,6,6,\allowbreak 8,6,6,8,8,4,6,6)$.

The nonclassical contribution still comes entirely from the first
S-box layer. The following rounds determine whether the favorable
first-round boundary remains useful.

\subsection{Post-diffusion boundaries of weight two}
\label{subsec:weight-two}

The minimum-weight boundary does not by itself establish global
optimality. Kiryukhin's analysis of non-minimum-weight LSX
differentials makes this point explicit~\cite{Kiryukhin2021}. We
therefore extended the two-round search to all post-$L_1$ boundaries of
byte weight two, while retaining the restriction that all $16$
first-layer external offsets are nonzero.

For a boundary $B$, let $\Delta'=L^{-1}(B)$ and define
$d^\star(\beta)=\max_\gamma D_S(\beta,\gamma)$. For
$w=\operatorname{wt}(B)\in\{1,2\}$, the corresponding score is
\begin{equation}
P_{2,c}(B)
=
q^{-(16+w)}
\left(
\prod_{i=0}^{15}m_c^\star(\Delta'_i)
\right)
\left(
\prod_{j:B_j\ne0}d^\star(B_j)
\right).
\label{eq:weight-two-score}
\end{equation}

There are $16\cdot255$ weight-one boundaries and
$\binom{16}{2}255^2=7{,}803{,}000$ weight-two boundaries. All were
generated and scored using~\eqref{eq:weight-two-score}. In the
classical case, $489{,}600$ boundaries contain a zero coordinate in
$L^{-1}(B)$ and hence have zero score under the all-nonzero-offset
restriction; the remaining $7{,}313{,}400$ are admissible. The maxima
are given in Table~\ref{tab:weight12-search}.

\begin{table}[t]
\centering
\begin{minipage}{0.72\textwidth}
\centering
\caption{Best two-round product-model characteristic scores for
post-$L_1$ byte weights one and two within the all-nonzero
first-layer-offset family.}
\label{tab:weight12-search}
\footnotesize
\begin{tabular*}{0.78\textwidth}
  {@{\extracolsep{\fill}}ccc@{}}
\toprule
$c$ & weight $1$ & weight $2$ \\
\midrule
\texttt{01} & $2^{-89.1504}$ & $2^{-93.3203}$ \\
\texttt{02} & $2^{-83.9801}$ & $2^{-85.9984}$ \\
\texttt{e1} & $2^{-83.9801}$ & $2^{-85.9984}$ \\
\texttt{04} & $2^{-89.9814}$ & $2^{-90.5739}$ \\
\texttt{03} & $2^{-91.1955}$ & $2^{-94.0480}$ \\
\bottomrule
\end{tabular*}
\end{minipage}
\end{table}

For the strongest class $\mathtt{02}/\mathtt{e1}$, no admissible
weight-two boundary in this restricted family improves on the
weight-one value $2^{-83.9801}$. This excludes the first non-minimum
post-diffusion weight class as an improvement within the stated family.
It does not prove global optimality over higher weights, nor does it
provide an MEDP bound.

\subsection{Relation with the $1\to16\to1$ topology}

The three-round family above has the activity pattern
$16\to1\to16$. Ishchukova et al.~\cite{Ish17} studied a different
Kuznyechik topology in which one active byte expands to sixteen and
later collapses to one, namely $1\to16\to1$.

The two families should not be compared by directly substituting one
reported probability for the other. The $1\to16\to1$ topology contains
$18$ active S-boxes, whereas the $16\to1\to16$ family used here contains
$33$. A meaningful comparison would require optimizing both the
classical and extended-differential constructions under the same admissibility
conditions. We do not perform that separate optimization here.

\subsection{Four-round behavior}

For four rounds we exhaustively reconstructed the pattern
$16\to1\to16\to1$. For every first post-$L_1$ byte position and value,
we enumerate every admissible second-round S-box transition; after the
second linear layer, we enumerate every weight-one post-$L_3$ boundary
and score the resulting sixteen ordinary third-round S-box transitions
and the final active S-box exactly. The best classical score is
$2^{-200.8105}$, while the best $c=\mathtt{02}$ score is
$2^{-199.2480}$. The inverse multiplier $c=\mathtt{e1}$ gives the same
optimum. Thus the extended-differential advantage decreases from the two- and
three-round values but remains about $1.56$ bits in this search family.

This behavior is consistent with Theorem~\ref{thm:first-interface}.
The multiplier changes only the first nonlinear transition. Every later
round is governed by ordinary differential propagation. Additional
rounds can therefore favor or suppress a particular first-round
boundary, but they do not introduce another extended-differential probability
factor.

The reduced-round results show that the unusually large Kuznyechik
extended-differential transitions produce higher product-model characteristic scores
than the corresponding classical characteristics at two and three rounds
within the restricted search families considered here. In the stated
four-round topology, the improvement is smaller but remains positive.

Extending the nonclassical relation itself over several rounds is a
different problem. Higher-order $c$-derivatives and differentials provide
one natural direction for such an extension~\cite{AKMRS23}. Whether they
lead to useful round-key cancellation or related-key constructions for
Kuznyechik requires a separate analysis and is left for future work.

\section{Exploratory long-round masked experiment}
\label{sec-cddt-model}
\label{sec:statistical-analysis}
\label{sec:Kuz1}
\label{sec:statistical-methodology-kuz}
\label{sec:Kuz2}

This experiment is deliberately separate from the exact two-round attack and the product-model characteristics above. Let $\widetilde E_K$ denote the Kuznyechik core without the initial whitening key. For a scalar $c$, plaintext $x$, and masked affine offset $a$, we evaluate the pair $x'=c\cdot x\oplus a$ through the selected number of rounds and record the projected XOR output difference $b=\pi_{M_O}(\widetilde E_K(x)\oplus\widetilde E_K(x'))$. The input mask $M_I$ restricts only $a$. Hence for $c\ne1$ a one-byte input mask does \emph{not} mean a one-byte plaintext XOR difference: $(c\oplus1)x\oplus a$ is generally dense. We therefore make no closed-form multi-round extended-differential characteristic claim for this experiment.

All reported experiments use one-byte input and output projections, so the masked cell space has $(2^8-1)2^8=65{,}280$ cells after zero affine offsets are discarded. For $c\ne1$, Theorem~\ref{thm:exact-masked-baseline} gives the random-permutation ensemble reference $p_0=1/65280$ for each joint cell. Conditional on the retained sample count $N_{\rm ret}$, a predeclared excess is evaluated by the exact binomial upper tail; FDR or FWER corrections refer only to the family actually included in that scan. Configuration-level G-tests are descriptive summaries. Adaptive thresholds, Fisher combinations, normality checks, SPRT traces, and the heuristic bias-decay curve produced by the implementation are treated only as exploratory diagnostics.

The tested multipliers were $\mathtt{01},\mathtt{02},\mathtt{03},\mathtt{04},\mathtt{91},\mathtt{be},\mathtt{e1}$, with $c=\mathtt{01}$ as the classical benchmark and the other values representing the exceptional inverse classes. The final corpus contains 36 mask configurations selected after preliminary screening; each configuration attempted $5\times10^6$ trials. Because the screening contributes to hypothesis selection, these data are not a predeclared confirmatory family. Full configuration tables, pseudocode, schematic figures, and raw program output are provided in the supplementary material.

The Python/SageMath implementation uses precomputed byte multiplication and $L$-layer tables and was checked against the RFC~7801 test vectors. Multiprocessing workers in the historical code seeded their Python generators independently from operating-system entropy using fresh 128-bit values. Thus labels such as ``seed 37'' or ``seed 42'' identify recorded runs but do not determine all worker streams. The historical count-$130$ nine-round run used 10 worker processes on the recorded M1-Max system; later retained reruns used higher-core systems, including 128- and 256-process executions. We therefore treat hardware and process counts as per-run provenance and the Monte Carlo outputs as recorded exploratory realizations rather than bit-for-bit replays from a single master seed.

In the historical nine-round run labelled seed 42, the $c=\mathtt{04}$ configuration \texttt{byte\_8\_in$\to$byte\_8\_out} contains the strongest cell $(a,b)=(\mathtt{29},\mathtt{8d})$ 130 times among $5{,}000{,}000$ attempted trials, corresponding to the reported $1.7\times$ bias. In that same historical run, a second $1.7\times$ configuration occurs for $c=\mathtt{e1}$, while at eight rounds a $c=\mathtt{03}$ configuration reaches $1.6\times$. Using the attempted-trial count conservatively, $X\sim\operatorname{Bin}(5{,}000{,}000,1/65280)$ gives $\Pr[X\ge130]=1.7597085465\times10^{-8}$. The corresponding Bonferroni values are $1.1487377391\times10^{-3}$ over one 65,280-cell configuration and $4.1354558610\times10^{-2}$ over $36\times65{,}280$ cells. The latter is \emph{not} a final family-wise $p$-value, because the 36 configurations followed preliminary screening whose full look-elsewhere family is not specified. The nine-round observation is therefore an exploratory fixed-key candidate requiring fresh predeclared validation.

Inverse multipliers must also be interpreted at the population level: Proposition~\ref{prop:masked-inverse} shows that, for fixed masks and permutation, $c$ and $c^{-1}$ differ only by the row relabelling $a\mapsto c^{-1}a$. Finite-run differences between, for example, $\mathtt{04}$ and $\mathtt{91}$ are therefore not evidence that one inverse representative is intrinsically stronger. At five through seven rounds no individual masked pair survives the reported FDR correction, whereas the one- and three-round classical scans show the expected much larger local departures. Detailed round-by-round tables are left to the supplement.

For planning a \emph{future predeclared} test, under the working alternative $p_1=1.7p_0$, exact binomial optimization at 80\% power gives approximately $1{,}075{,}055$ retained pairs for one preselected cell, $5{,}316{,}433$ when Bonferroni-correcting 65,280 cells, and $6{,}627{,}964$ for $36\times65{,}280$ cells. These are prospective design calculations, not retroactive corrections of the screened experiment.

The theorem-level cryptanalytic result of this paper remains the exact two-round master-key recovery of Section~\ref{sec:standard-2r-recovery}; the long-round experiment is included only to document a reproducible direction for future predeclared study.

\section{Conclusions}
\label{conc}

The exceptional extended-differential behavior of the Kuznyechik S-box is explained by the interaction between its hidden pseudo-exponential representation and the public Kuznyechik field. The cross-field transfer theorem and exact degree--rank result yield structural lower bounds for the classes $\mathtt{02}/\mathtt{e1}$, $\mathtt{04}/\mathtt{91}$, and $\mathtt{03}/\mathtt{be}$. The complete-table identities then quantify how this structure is distributed over offsets, while whitening turns the row index into the key-dependent value $a\oplus(1\oplus c)K$.

For two-round standard Kuznyechik, the exact channel $Q_{c,\ell}=W_cP_\ell W_1$ retains enough row separation for rigorous maximum-likelihood recovery. With $c=\mathtt{04}$, both $128$-bit master-key halves are recovered with failure probability at most $0.01$ using $226{,}512$ chosen pairs and at most $240{,}669$ encryption queries. Restricted product-model searches separately show gains of $5.2$, $4.6$, and about $1.6$ bits at two, three, and four rounds; exhaustive post-$L_1$ weight-two enumeration within the all-nonzero
first-layer-offset family does not improve the best $\mathtt{02}/\mathtt{e1}$ weight-one boundary. These characteristic scores are not substitutes for exact fixed-key probabilities.

The screened long-round Monte Carlo study addresses a different question and is intentionally reported as exploratory. A natural next step is a predeclared validation experiment, while the theoretical direction is to incorporate the extended first-interface cost into a broader Kiryukhin-style optimization or to retain the exact joint nonlinear/key distribution beyond two rounds.

\section*{Supplementary material and data availability}
A separate supplementary document contains the full Kuznyechik technical specification, the complete extended and outer differential-uniformity tables, detailed round-by-round cryptanalytic tables, and representative statistical output. Complete configuration files, computation logs, implementation details, and verification outputs are additionally provided in the reproducibility repository~\cite{github}.


\end{document}